\documentclass[11pt,a4paper,UKenglish]{article}

\usepackage{geometry}
\usepackage{amsthm}
\usepackage{amsmath}
\usepackage{amsfonts}
\usepackage{graphicx}

\newcommand{\eps}{\varepsilon}

\newcommand{\A}{\mathcal{A}}

\newcommand{\coeff}[2]{{\operatorname{coeff}_{#2}(#1)}}

\def\sign{\operatorname{sign}}

\def\val{\operatorname{val}}

\newtheorem{theorem}{Theorem}
\newtheorem{conjecture}[theorem]{Conjecture}
\newtheorem{claim}[theorem]{Claim}
\newtheorem{definition}[theorem]{Definition}
\newtheorem{lemma}[theorem]{Lemma}

\def\R{\mathbb R}
\def\eps{\varepsilon}

\def\mynorm{{2,\infty}}

\renewcommand{\cos}{\operatorname{cos}}
\renewcommand{\sin}{\operatorname{sin}}
\title{Paraunitary Matrices, Entropy, Algebraic  Condition Number  and Fourier Computation
\footnote{This work was supported by a consolidator ERC grant INF-COMP-TRADEOFF}}
\author{Nir Ailon, Technion IIT}
\begin{document}

\maketitle
\def\spectral{{\operatorname{spec}}}
\def\sign{\operatorname{sgn}}
\def\dim{n}
\def\KL{{\operatorname{KL}}}
\def\C{\mathbb C}
\def\R{\mathbb R}
\def\P{\mathcal P}
\def\Q{\mathcal Q}
\def\trace{\operatorname{tr}}
\def\diag{\operatorname{diag}}
\def\rank{\operatorname{rank}}
\def\F{{\mathcal F}}
\def\Id{\operatorname{Id}}
\def\f{{\hat f}}
\def\Z{{\mathbb Z}}
\def\X{{\cal X}}
\def\B{{\cal B}}
\def\Ellips{{\cal E}}
\def\S{{\cal S}}
\def\err{{\operatorname{err}}}
\def\MLAE{{\operatorname{MLAE}}}
\def\N{{\mathcal N}}
\def\tr{\operatorname{tr}}
\def\consts{\mathcal{C}}
\def\ring{\mathcal{R}}
\def\Im{{\operatorname{Im}}}
\def\DFT{\operatorname{DFT}}
\def\conv{\operatorname{CONV}}
\def\poly{\operatorname{poly}}

\setcounter{page}{0}
\begin{abstract}
The Fourier Transform is one of the most important linear transformations used in science and engineering.  Cooley and Tukey's Fast Fourier Transform (FFT) from 1964 is a method for computing this transformation in time $O(n\log n)$.   From a lower bound perspective,  relatively little is known.  Ailon shows in 2013 an $\Omega(n\log n)$ bound for computing the normalized Fourier Transform assuming only unitary operations on two coordinates are allowed at each step, and no extra memory is allowed.  In 2014, Ailon then improved the result to show that, in a $\kappa$-well conditioned computation, Fourier computation can be sped up by no more than $O(\kappa)$.   The main conjecture is that Ailon's result can be exponentially improved, in the sense that $\kappa$-well condition cannot  admit $\omega(\log \kappa)$ speedup.  

The main result here is that `algebraic' $\kappa$-well condition cannot admit $\omega(\sqrt \kappa)$ speedup.  
One equivalent definition of algebraic condition number 
is related to the degree of polynomials naturally arising as the computation evolves.  Using the maximum modulus theorem from complex analysis, we show that algebraic condition number upper bounds standard condition number, and equals it in certain cases.    Algebraic condition number is an interesting measure of numerical computation stability in its own right, and provides a novel computational lens.  Moreover, based on  evidence from other recent  related work,  we  believe that the approach of algebraic condition number has a good chance of establishing an algebraic version of the main conjecture.
\end{abstract}

\section{Introduction}
The (discrete) normalized Fourier transform is a complex linear mapping sending an input $x\in \C^n$ to $y=Fx\in \C^n$, where $F$ is an $n\times n$ unitary matrix defined by
\begin{equation}\label{dft} F(k,\ell) = n^{-1/2}e^{-i2\pi k\ell/n}\ .
\end{equation}
The Fast Fourier Transform (FFT) of Cooley and Tukey \cite{CooleyT64} is a method for computing the Fourier transform 
of a vector $x\in \C^n$
in time $O(n\log n)$ using a so called linear-algebraic algorithm.  A linear-algebraic algorithm's state at each point in the computation is a vector in $\C^n$, representing a linear transformation of the input.  Each coordinate of the vector is a complex number stored in a memory location (or rather, two real numbers, one representing the real part and the other the imaginary part).  The next state is obtained from the current one by a simple linear algebraic transformation.   We will define this precisely in what follows.

As for lower bounds, it is trivial  that computing the Fourier Transform requires a linear  number of steps, because each 
coordinate in the output depends on each coordinate in the input.

There has not been much prior work on better bounds, for Fourier transform over the real or complex field,
with respect to a reasonable model of computation.

In 1973, Morgenstern proved that if the modulus of constants used in an unnormalized Fourier transform algorithm are bounded by a global constant,
then the number of steps required  is at least 
$\Omega( n\log n)$.  He used a potential function related to matrix determinant.  However, this result does not provide lower bounds for  normalized Fourier transform, which can be computed by running the  unnormalized FFT (as described in
\cite{CooleyT64}), dividing all the
constants by the global constant $\sqrt 2$.  
Pan \cite{PAN198611} shows an $\Omega(n \log n)$ bound for so-called \emph{synchronized} algorithms which are, roughly speaking,
based on a layered computational graph (as in the well-known FFT).  It is not clear why synchronicity is required for optimality.

Papadimitriou  proves in \cite{Papadimitriou:1979:OFF:322108.322118} that the graph structure
of Cooley and Tukey's FFT algorithm is optimal if we are only allowed to multiply by $n$'th roots of unity constants,
and the transformation is computed in a finite field.  It should be noted that removing the finite field requirement from
Papadimitriou's paper results in a lower bound that is subsumed by that of Morgenstern's, and in particular, applies
only to the unnormalized FT.  Again, it is not clear how to obtain
something for a normalized FFT from this result.

 Winograd  \cite{Winograd76} was able to reduce by $20\%$  the number of multiplications in Cooley and Tukey's FFT, without changing the number of additions.
In \cite{Winograd76b} the same author shows that at least linearly many multiplications are needed for performing Cooley and Tukey's algorithm for computing the \emph{unnormalized} FFT, defined as $\sqrt n$
times (\ref{dft}).

We also note that in the quantum  world, an $O(\log^2 n)$-time algorithm is known using quantum gates, each  unitary acting on only $4$ entries.  Refer to \cite{Hales00animproved} and to references within for this line of work, which is not further discussed here.

 Ailon   \cite{Ailon13}    considered a computational model  allowing only $2\times 2$ rotation gates acting on a pair of coordinates. He showed an $\Omega(n \log n)$ lower bound on the number of such gates required for the normalized Fourier transform.  The proof was done by defining a potential function on the matrices $M^{(t)}$ defined as the transformation that takes the input to the machine state in step $t$.  The potential
function is simply the sum of Shannon entropy of the probability distributions defined by the squared modulus of elements in the
matrix rows.  (Due to orthogonality, each row, in fact,  defines a probability distribution).  That work raised the question of \emph{why shouldn't it be possible to gain speed by `escaping' the unitary group?}
 In \cite{DBLP:journals/toct/Ailon16} he added nonzero, constant multiplication gates acting on a single coordinate to the model.
 By generalizing the entropy potential function to the  \emph{quasi-entropy} function, so called because it is applied  to possibly negative numbers,
the  approach allows proving that a speedup in computing \emph{any} scaling of FFT \emph{requires} either working
  with very large or very small  numbers (asymptotically growing accuracy).   The bounds are expressed using a quantitative relationship between time (number of  gates) versus the standard notion  of  condition number of matrices.  More precisely, it is shown that speedup of FFT by factor of $b$ implies that at least one $M^{(t)}$ must be $\Omega(b)$-ill-conditioned.  In \cite{DBLP:conf/icalp/Ailon15}, Ailon further showed that speedup implies ill-condition at \emph{many} points in the computation, affecting independent information in the input. 
Ailon's results \cite{Ailon13, DBLP:journals/toct/Ailon16, DBLP:conf/icalp/Ailon15}  are believed not to be tight.
The main open question, informally stated, is as follows. We will restate it as Conjecture~\ref{conj2} in a precise way in Section~\ref{evidence}, after building our framework.
\begin{conjecture}\label{conj1}[Informal]
Any speedup of FFT (counting linear algebraic operations), if existed, would require `exponential (in the speedup) ill-condition'.
\end{conjecture}

\noindent
Formal 
versions of Conjecture~\ref{conj1} are given in Conjectures~\ref{conj2} and~\ref{conj3} below.

\subsection{Our Contribution}
This work defines a new  and natural \emph{algebraic} notion of matrix condition number, which upper bounds the usual definition 
of condition number, and equals it in some cases.  Using this new notion, we are able to show that  $\kappa$-well conditioned computation allows only $O(\sqrt\kappa)$ speedup of FFT, in other words, a quadratic improvement on the previous result, albeit with a weaker notion of condition. More importantly, we believe that an algebraic version
of Conjecture~\ref{conj1} is achievable with respect to algebraic condition, and justify this belief in Section~\ref{evidence}.

\subsection{\underline{The RAM Computational Model and This Work }}
Before we continue, we shall emphasize a point about a popular computational model:  \emph{RAM}. Roughly speaking, in this model, a computer word contains $\Theta(\log n)$ bits of precision, for input of size $n$, and each addition/multiplication is counted as $1$ operation.  While highly suitable for problems of complexity that allows ``hiding poly-log factors",  this model should be used with caution for lower bounding FT computation.  One of the main arguments often stated against Conjecture~1 is that poly$(n)$ ill-condition, in the extreme case of $\Theta(\log n)$ speedup of FFT, roughly means $\Theta(\log n)$ bit words, which is assumed by the RAM model anyway, so why is the conjecture interesting!?  
But why do we assume the RAM model?
It is tempting to think that we must assume it, due to the following example: If we were to allow the more realistic $O(1)$ bits per word,
and compute the normalized Fourier transform of $(1,1,1,....,1)$, then the result would have a coordinate equalling $\sqrt n$ and
hence  would overflow.  But notice that all the other coordinates of the output equal $0$.  Requiring $\Theta(\log n)$
bits per word in all words is an overkill that merely satisfies an instinct to work with equal, fixed sized words.  

For readers who insist on a RAM-like model:  Our line of work studies the complexity of computing the FT for inputs
in $n$ dimensional real space, containing $O(1)$ bits per word \emph{on average} (over the coordinates).  Such representation allows describing a point inside
the Euclidean ball of radius $\sqrt n$, to within fixed accuracy per coordinate.  The size of the word can vary with the computation, 
according to the number stored in the word.  Additionally, the cost of performing each operation grows with the number of bits
involved in the operands.  The main argument is that speedup of FFT requires a growing (in the speedup)
number of bits in the words participating in the algorithm steps, and to quantify this growth.

Alternatively, the reader is invited to think about well-conditioned linear computation as a computational model worthy of discussion in its own right, and will hopefully be convinced that it is more suitable to think about FT complexity in such a model, rather than under a RAM or  a  RAM\emph{ish} lens.

\section{Fourier Transform Types}
In theory and practice of  engineering and computer science, there are two main types of Fourier transforms considered.
The Discrete Fourier Transform (DFT) in $n$ dimensions, as defined in (\ref{dft}), is a complex unitary mapping defined by the characters of the  Abelian group $\Z/n\Z$.  The Walsh-Hadamard transform is a real orthogonal mapping defined by
the characters of the $n$ dimensional binary hypercube $(\Z/2\Z)^n$.  More precisely, for $n$ an integer power of 2,
the $(i,j)$'th matrix element
 is defined as $\frac 1 {\sqrt n}  (-1)^{\langle [i-1], [j-1]\rangle}$, where $\langle\cdot,\cdot\rangle$ is dot-product, and $[p]$ denotes the bit
representation of the integer $p\in \{0,\dots, n-1\}$ as a vector of $\log n$ bits.  Similarly to FFT for DFT, there is an $O(n\log n)$ algorithm for computing the Walsh-Hadamard transform of an input $x$.

  Throughout, we will assume $n$ is an integer power of $2$ and will use $F$ to denote either the $n-$Walsh-Hadamard, which is a real orthogonal transformation, or the real representation of the  $(n/2)$-DFT, which is a
  complex unitary transformation.   By \emph{real representation}, we mean the standard view of  a complex
  number $a+\iota b$ as a column vector $(a,b)^T$ in two dimensional real space, and the multiplicative action $a+\iota b \mapsto (c+\iota d)(a+\iota b)$ given by left-multiplication by the matrix $\left (\begin{matrix} c & -d \\ d & c \end{matrix}\right)$.  In that way, a vector in $n$ dimensional complex space embeds to a $2n$ dimensional real space, and a $n\times m$ complex matrix embeds as a $(2n)\times (2m)$ real matrix.
This allows us to avoid defining a computational model over the complex field, because such computation can be emulated by reals anyway.  Nevertheless, our analysis will require extension to the complex field.
  
\section{Computational Model and Notation}\label{sec:modelnotation}



We remind the reader of the computational model discussed in \cite{Ailon13,DBLP:journals/toct/Ailon16}, which is a special case of
the linear computational model.  The machine state represents a vector in $\R^\ell$ for some $\ell\geq n$,
where initially it equals the input $x\in \R^n$ (with possible padding by zeroes, in case $\ell>n$).  Each step (gate)
is either a \emph{rotation} or a \emph{constant}.   

A rotation applies a $2\times 2$ orthogonal operator mapping a pair of
machine state coordinates (rewriting the result of the mapping to the two coordinates).  
Note that a rotation includes the $2\times 2$ mapping affecting a single coordinate, e.g.
$\left ( \begin{matrix} -1 & 0 \\ 0 & 1 \\ \end{matrix} \right )$.  (Technically this is a reflection, but we
refer to it as a \emph{rotation gate} here.)

A constant gate multiplies a single machine state coordinate (rewriting the result) by
a positive real constant.  
In case $\ell=n$ we say that we are in the in-place model.   Any nonsingular linear mapping over $\R^n$ can be decomposed into a sequence of rotation and constant gates in the in-place model, and hence our model is universal.  Normalized FFT  works in the in-place model, using rotations only.  A restricted method for dealing
 with $\ell>n$ was developed in \cite{DBLP:journals/toct/Ailon16}. We focus in this work on the in-place model only.

Since both rotations and constants apply a linear transformation on the machine state, their composition is a linear transformation.  If $\A_n$ is an in-place algorithm for computing a  linear mapping over $\R^n$, it
is convenient to write it as $\A_n = (M^{(0)}=\Id, M^{(1)}, \dots, M^{(m)})$ where $m$ is the number of steps,
$M^{(t)}\in \R^{n\times n}$ is the mapping that satisfies that for input $x\in \R^n$ (the initial machine state),
$M^{(t)}x$ is the machine state after $t$ steps.  ($\Id$ is the identity matrix).  The matrix $M^{(m)}$ is the target
transformation, which will typically be  $F$ in our setting.  For $t\in[m] := \{1,2,\dots, m\}$, if  the $t$'th gate is a rotation, then $M^{(t)}$ differs
from $M^{(t-1)}$ in at most two rows, and if the $t$'th gate is a constant, then $M^{(t)}$ defers from $M^{(t-1)}$
in at most one row.  

Throughout, for a matrix $M$, we let
$M_{i,:}$ denote the $i$'th row of $M$.  The symbol $\iota$ denotes $\sqrt{ -1}$.  For an integer $a$, notation $[a]$ is shorthand for the set $\{1,\dots, a\}$.  All logarithms are assumed to be base $2$, unless the base is explicit, as in $\log_5 6$.  For a  matrix $M$ over the field of real or complex numbers, $\|M\|$ is the spectral norm.

\section{Replacing constant gates with an  indeterminate}
Fix an in-place algorithm $\A=(M^{(0)}=\Id,M^{(1)}, \dots, M^{(m)} = F)$ computing $F$.
Let $\consts_\A$ denote the set of values  used in the constant gates.
Let $0<\Delta< 1$ be some real number.  
If each $c\in \consts_\A$ is an integral power of $\Delta$, then we say that $\A$ is integral with respect to $\Delta$.
 We will assume that $\Delta \geq  2/3$ (otherwise we could replace $\Delta$
with $\Delta^{1/\ell}$ for a suitable integer $\ell$, with respect to which $\A$ is still integral).

Assume that  $\A$ is integral with respect to some $2/3 \leq\Delta< 1$.   The integrality assumption will be removed in Section~\ref{sec:def:algcond} using limit arguments.
 We define a new algorithm $\A_\Delta$ with
corresponding matrices $(M^{(0)}_\Delta,\dots, M^{(m)}_\Delta)$ where each $M^{(t)}_\Delta$ is now an $n\times n$
matrix over the ring of Laurent polynomials with complex  coefficients  in an indeterminate $z$.\footnote{We remind the reader that a Laurent polynomial possibly has negative degree monomials.} (Throughout, by \emph{polynomials} we refer to Laurent polynomials, and by \emph{proper polynomials} we mean that only nonnegative power monomials are allowed). In what follows, we may interchangeably  write $p$ or $p[z]$ if it is clear from the context that
$p$ is a (Laurent) polynomial in $z$.

To get accustomed  to our notation, note that for some matrix $A=A[z]$ in this ring we can write e.g. $A_{1,3}[7]$,  referring to the (complex) number in the first row, third column
of the matrix obtained by assigning the value $7$ to $z$ in $M^{(t)}_\Delta$.
We now inductively define $M^{(t)}_\Delta$.    The first matrix $M^{(0)}_\Delta$ is simply $M^{(0)} = \Id$.
Having defined $M^{(t-1)}_\Delta$, we define $M^{(t)}_\Delta$ depending on whether the $t$'th gate is a rotation,  or a constant gate.  In case of rotation, the matrix $M^{(t)}_\Delta$ is obtained from $M^{(t-1)}_\Delta$ by performing the same transformation  giving $M^{(t)}$ from $M^{(t-1)}$.  
 More precisely, if $M^{(t)}$ is obtained from $M^{(t-1)}$ by a  rotation matrix $R$,  namely
 $M^{(t)} = R M^{(t-1)}$, then
 $M^{(t)}_\Delta = R M^{(t-1)}_\Delta$.

In case of a constant gate with value $0<c\in \consts_\A$, we  replace the value $c$ with the monomial $z^{\log_\Delta c}$.  More precisely, if $M^{(t)}$ is obtained from $M^{(t-1)}$ by left-multiplication
with the matrix $\diag(1,\dots, 1, c, 1,\dots, 1)$, then
$$ M^{(t)}_\Delta = \diag(1,\dots, 1, z^{\log_\Delta c}, 1, \dots, 1)M^{(t-1)}_\Delta\ .$$

By the integrality of $\A$ with respect to $\Delta$, the number $\log_\Delta c$ is indeed an integer.  
From the definition of $M^{(t)}_\Delta[z]$,  a simple induction implies that  for all $t\in \{0\dots m\}$,
$M^{(t)}_\Delta[\Delta] = M^{(t)}$.
In words, this means that the $t$'th matrix in the original algorithm $\A$ is obtained by polynomial evaluation at $\Delta$.

Throughout, for a polynomial $p=p[z]$, we let $\coeff{p}{i}$ denote the coefficient (matrix, vector or scalar)
of $p$ corresponding to $z^i$.  In particular, $p = \sum_i \coeff{p}{i} z^i$. \footnote{We avoid subscripting for denoting polynomial coefficients, because that would be confused with matrix or vector indexing.}

\subsection{Paraunitary Matrices and Algebraic Condition Number}

Here and throughout, for a complex number $u$, $\bar u$ is its complex conjugate.  For a complex
polynomial $p = p[z]$, complex conjugation $\bar p$ is the polynomial obtained by complex-conjugating the coefficients of $p$
and substituting $z^{-1}$ for $z$.
For a polynomial matrix $M=M[z]$, the matrix $M^*$ is obtained by transposing $M$ and applying \emph{polynomial} complex conjugation (as just defined) to all its elements.
 For a nonzero polynomial $p[z]$, let $\deg(p)$ denote the maximal degree of $z$ (with a nonzero coefficient), and let  $\val(p)$  denote the minimal degree.   (The operator name $\val$ is short for \emph{valuation}, a standard algebraic abstraction for which ``min degree'' is the most common special case.)  
For example, if $p[z] = -6z^{-2}+3+z^5$, then $\deg(p)=5, \val(p)=-2$.  
Note that $\deg(p)$ may be negative, and $\val(p)$ may be positive.
For a polynomial matrix $M[z]$,  $\deg(M) = \max_{i,j} \deg(M_{i,j})$
and $\val(M) = \min_{i,j} \val(M_{i,j})$.

\begin{definition}\label{dfn:paraunitary}
A complex polynomial matrix $U[z]$
is \emph{paraunitary} if
$ U^* U = \Id.$
\end{definition}
Paraunitary matrices are used in signal processing 
(see e.g. \cite{paraunitary2}, chapter 14), 
and  have some useful
properties for our purposes.  First, they
  form a multiplicative group.  This fact is useful for verifying the following:
\begin{claim}\label{clm:allparaunitary}
The  matrices $M^{(t)}_\Delta$ are paraunitary for all $t$.
\end{claim}
\noindent
The following is easy to see from the definitions, using simple induction:
\begin{claim}\label{clm:paraunitaryunitary}
For paraunitary $U[z]$, the evaluation $U[\omega]$ is a (complex) unitary matrix for any $\omega$
a complex number on the complex unit circle. 
\end{claim}
 This will be used below.
We now remind the reader of the   definition of a  $\kappa$-well conditioned algorithm \cite{DBLP:conf/icalp/Ailon15,DBLP:journals/toct/Ailon16}.  In this work, we will refer to
this standard notion of well conditionedness  \emph{geometric}, because it is related to geometric stretching and shrinking
of vectors under linear operators.
\begin{definition}
The (geometric) condition number of a  nonsingular complex matrix $M$ is \\ 
$\|M\|\cdot\|M^{-1}\|$.
A nonsingular  matrix $M$ is (geometrically) $\kappa$-well conditioned if  its (geometric) condition number is
at most $\kappa$.  Otherwise it is (geometrically) $\kappa$-ill conditioned.
The algorithm $\A=(\Id,M^{(1)}\dots M^{(m)})$ is (geometrically) $\kappa$-well conditioned  if $M^{(t)}$ is (geometrically) well conditioned for all $t=0\dots m$.  Otherwise, it is (geometrically) $\kappa$-ill conditioned.  The (geometric) condition number of $\A$ is the smallest $\kappa$ such that $\A$ is (geometrically) $\kappa$-well conditioned.
\end{definition}

We define a different, though not unrelated (as we shall see below) notion of condition number of an algorithm.  We first define it for $\Delta$-integral algorithms, and then extend to general algorithms.
\begin{definition}\label{defn:algwellcond}
A $\Delta$-integral algorithm $\A$ is $\Delta$-\emph{algebraically} $\kappa$-well conditioned if for all $t=1\dots m$, 
$ \Delta^{\val(M^{(t)}_\Delta)-\deg(M^{(t)}_\Delta)} \leq \kappa$.
Otherwise, it is $\Delta$-algebraically $\kappa$-ill conditioned.
\end{definition}
Although the definition seems to depend on it, the parameter $\Delta$ is immaterial for defining
(and thinking about) algebraic condition number. It is trivial to see that if $\A$ is integral with respect 
to both $\Delta$ and $\Delta'$ and is $\Delta$-algebraically $\kappa$-well conditioned, 
then it is also  $\Delta'$-algebraically $\kappa$-well conditioned.  
  In fact, even if an algorithm is integral with respect to \emph{no} $\Delta$, 
we can extend the definition by using limits.  Being algebraically $\kappa$-well conditioned is henceforth
a property of any algorithm $\A$, regardless of integrality.  The precise definitions are given in Section~\ref{sec:def:algcond}, for completeness.  We   henceforth define:
\begin{definition}\label{defn:algcondnum}
The algebraic condition number of $\A$ is the smallest $\kappa$ such that $\A$ is algebraically $\kappa$-well conditioned.
\end{definition}

If we consider some intermediate matrix $M^{(t)}$ given rise to by our  algorithm $\A$,
then algebraic  well conditioning cannot be determined from $M^{(t)}$ `in vitro', because it depends
on \emph{how} we reached $M^{(t)}$ from the initial state $\Id$.  This is in contrast with the notion of (geometric) condition, which can be determined from a single matrix.
Nevertheless, the  following lemma tells us that the two notions are quantitatively related.
\begin{lemma}\label{lem:condition}
If an algorithm $\A$ is algebraically $\kappa$-well conditioned, then it is
geometrically $\kappa$-well conditioned.  There exist algorithms for which the two
notions of condition number coincide.
\end{lemma}
The proof of this connection requires some complex analysis, and in particular
the maximum modulus theorem.
We provide it here only for the case $\A$ is integral with respect to some $\Delta$.  The extension to general $\A$ is not difficult using limit arguments, and we present it in Section~\ref{sec:lemgen}.
\begin{proof}
Assume $\A$ is integral with respect to $\frac 2 3 \leq \Delta<1$, and fix $t\in[m]$.  
Let $v = \val(M^{(t)}_\Delta)$ and $d = \deg(M^{(t)}_\Delta)$.
By definition of $\val$ and $\deg$, the polynomial matrix $M^v[z] := z^{-v} M^{(t)}_\Delta$
is a proper polynomial matrix (with no negative degree monomials).  Therefore, it is an entire
(matrix valued) function of $z$, in the complex sense.  By the (weak) maximum modulus theorem for Banach
space valued functions on the complex plane (\cite{DunfordS58}, page 230), we conclude that 
$ \|M^v[\Delta]\| \leq \max_{\omega\in \C : |\omega|=1} \left \| M^v[\omega] \right \|\ ,$
because $\Delta$ is contained in the open unit disk.
The right hand side is identically $1$ by Claims~\ref{clm:allparaunitary} and~\ref{clm:paraunitaryunitary}.  The left hand side equals exactly $\Delta^{-v} \|M^{(t)}_\Delta[\Delta]\| = \Delta^{-v} \|M^{(t)}\|$.  Therefore,
\begin{equation}\label{eq:normineq1}
\|M^{(t)}\| \leq \Delta^{v}\ .
\end{equation}

\noindent
By Claim~\ref{clm:allparaunitary} and Definition~\ref{dfn:paraunitary}, the inverse of the matrix $M^{(t)}_\Delta$ is $(M^{(t)}_\Delta)^*$.\footnote{The matrix $M^{(t)}_\Delta$ is real, and hence its transpose  equals its conjugate-transpose.  We prefer to keep the complex notation $(\cdot)^*$.)}
    Let $$M^d := z^{d} (M^{(t)}_\Delta)^*\ .$$
By construction, the polynomial matrix $M^d$ is a proper polynomial, and hence again by
the (weak) maximum modulus theorem,
$\Delta^{d} \|   (M^{(t)}_\Delta)^*[\Delta] \| = \|M^d[\Delta]\| \leq \max_{\omega\in \C : |\omega|=1} \left \| M^d[\omega] \right \| = 1\ .$
But by paraunitarity $(M^{(t)}_\Delta)^*$ is the inverse of $M^{(t)}_\Delta$, and this holds
in particular for $z=\Delta$.  Hence the last inequality gives $\Delta^{d} \|(M^{(t)})^{-1}\| \leq 1$ or:
\begin{equation}\label{eq:normineq2} \|(M^{(t)})^{-1}\| \leq \Delta^{-d}\ .
\end{equation}
Combining (\ref{eq:normineq1}) with (\ref{eq:normineq2}), we get
$ \|(M^{(t)})^{-1}\|\cdot  \|M^{(t)}\| \leq \Delta^{v}\cdot \Delta^{-d} \leq \kappa\ , $
where the rightmost inequality is by the lemma's assumption.
This concludes the first part of the result.

To show tightness, it is enough to  work with $2\times 2$ matrices.   Consider an algorithm $\A$ performing
two steps.  The first multiplies the first coordinate by a  positive number $\Delta$, and the second multiplies the
second coordinate by $\Delta^{-1}$.  It is easy to see that both the algebraic and the standard notions
of condition number coincide for this case.
\end{proof}

\noindent
We are now ready to state our main result.

\begin{theorem}\label{thm:main}
Assume algorithm $\A=(M^{(0)},\dots, M^{(m)})$ computes $F$, 
and is algebraically $\kappa$-well conditioned.  Then 
\begin{equation}\label{gugugu} m=\Omega\left(\frac{ n \log n}{\sqrt \kappa}\right)\ .\end{equation}
\end{theorem}

The proof is presented in the remainder of the paper, for the case $\A$ is integral with respect to some
$\frac 2 3 \leq \Delta \leq1$.   
The general case,  using limit arguments, is deferred to Section~\ref{sec:thmgen}.
But first,

\subsection{ Algebraic-Condition Number: A Novel Computational Lens}
Geometric condition number is a standard matrix invariant that appears in classic textbooks on matrix theory,
and is a useful invariant when analyzing convergence rates of certain numerical algorithms on matrices.
It is important to discuss whether the bound of Theorem~\ref{thm:main} is interesting under the non-standard
algebraic condition number notion.  Admittedly, we were not able to prove a version of Theorem~\ref{thm:main}
with `geometrically $\kappa$-well conditioned' replacing `algebraically $\kappa$-well conditioned'.  However, we
have good reason to believe that there is a chance of proving an algebraic version of Conjecture~\ref{conj1}, in which
the denominator of (\ref{gugugu}) is logarithmic in $\kappa$
(see Conjecture~\ref{conj2} below for an exact statement).
Evidence for this is presented in Section~\ref{evidence}. 
Proof of Conjecture~\ref{conj2} would be an important step toward understanding the complexity of computing 
the Fourier transform (and perhaps other problems).
Additionally, proof of the algebraic conjecture would
 highlight the importance of its (stronger) geometric counterpart.

One issue is, that it may seem unnatural to introduce polynomials for a problem involving only real arithmetic.
But there is a simple alternative way to think about algebraic condition number that does not involve polynomials: 
If we think of the constant gates as elements in the multiplicative group $\R_{>0}$, and the coefficients of the rotation and reflection
operations as elements of the field $\R$, then the machine state at each step $t$ is now a vector with coefficients
in the group algebra $\R[\R_{>0}]$. Accordingly, the matrices $M^{(t)}$ are also defined over $\R[\R_{>0}]$.  The algebraic
condition number of a matrix is now the ratio between the largest group element $c\in \R_{>0}$ with nonzero coefficient
in some matrix element, and the smallest one.  Intuitively, a large  (resp. small) group element  $c\in \R_{>0}$ occurs in the computation
as a result of a chain of constant gates, the product of which is $c$, that affect  some direction in the input $n$-dimensional vector space.

Note that numbers that are introduced in the computation from  rotations (trigonometric functions of rotation angles), however miniscule,
do not affect the 
algebraic condition number, because they are used as coefficients in the group algebra $\R$, and not as elements in the group $\R_{>0}$. This property
is shared with geometric condition number, which is also not affected by rotation action.


\subsection{Polynomial Matrices: Norms and  Entropy}

We will need to define a norm for polynomials, polynomial vectors and polynomial matrices, over the
complex field.  
For a polynomial $p=p[z] \in \C[z]$, we define 
$ \|p\| = \sqrt{\sum_k |\coeff{p}{k}|^2}\ .$
The following is well known from harmonic analysis:
\begin{equation}\label{parseval} \|p\| = \sqrt{\int_{0}^1 \left |p\left [  e^{2\pi \iota t}\right ] \right |^2 dt}\ .
\end{equation}
For a polynomial vector $v=v[z] \in \C[z]^n$ we define
\begin{equation}\label{defpolynorm} \|v\| = \sqrt{\sum_{i=1}^n \| v_i\|^2}\ ,\end{equation}
where we remind the reader that $v_i$ is the $i$'th coordinate of $v$.
Finally, for a polynomial matrix $A$, we define $\|A\|_\mynorm$ as shorthand for $\max_i \|A_{i,:}\|$.
Let $M$ denote a real paraunitary matrix. Then we extend the definition of the quasi-entropy potential from \cite{DBLP:journals/toct/Ailon16,DBLP:conf/icalp/Ailon15} as follows:
\begin{equation}\label{def:entropy} \Phi(M) = \sum_k \sum_{i,j}\ h\left (\left | \coeff{M_{i,j}}{k}\right |^2 \right ) 
\end{equation}
where $h : \C \mapsto \C$ is defined by $h(u) = -u\log|u|$.  (Note:  No, this is not a mistake, and we do not want to define $h$ as $-|u|\log|u|$.  
We will be using $h$ for possibly negative and even complex values of $u$, where this makes a difference.\footnote{This comment will be removed from published versions.})
For the Fourier transform $F$, 
\begin{equation}\label{fourierentropy} \Phi(F) = \Theta(n\log n)\ .\end{equation}

We also define a preconditioned version of $\Phi$ for any paraunitary  matrix $M$, parameterized
by two polynomial matrices $A,B$:
\begin{equation}\label{yokyok} \Phi_{A,B}(M) = \sum_k\sum_{i,j}  h\left ( \coeff{(MA)_{i,j}}{k}\overline{\coeff{(MB)_{i,j}}{k}}\right ) .\end{equation}
where the outer sum is over $k$ for which the internal sum is not null.
Note that $\Phi_{\Id, \Id}(M) = \Phi(M)$.
The general idea of a preconditioned quasi-entropy potential was developed in \cite{DBLP:conf/icalp/Ailon15}.  The  idea there (and here) is to carefully choose fixed $A$ and $B$, and
to track the potential along the evolution of $M$ under algorithm steps.
The following key lemma bounds the absolute value of the change in entropy of a paraunitary matrix,
incurred by multiplication of rows by monomials (corresponding to constant multiplication in the original algorithm) and by planar rotations (corresponding to rotations in the original algorithm).
\begin{lemma}\label{lem:entropychange}
(1) Let $M[z]$ be a paraunitary matrix, and let $A[z],B[z]$ be some complex matrices.
Let $M'[z]$ be a paraunitary matrix obtained  from $M[z]$ by multiplication by a diagonal
matrix with monomials on the diagonal, that is
$ M'[z] = \diag(z^{i_1},\dots, z^{i_n}) M[z]\ .$
Then $\Phi_{A,B}(M') = \Phi_{A,B}(M)\ .$
(2)If $M'$ is obtained from $M$ by a planar rotation action, then
$$ \left | \Phi_{A,B}(M') - \Phi_{A,B}(M)\right| = O\left ( \|{MA}\|_\mynorm\cdot \|{MB}\|_\mynorm\right )\ .$$
\end{lemma}

\begin{proof}
Part (1) is trivial, from the definitions.  
For part (2), polynomial matrix potential function is, equivalently, a scalar matrix potential function (defined in \cite{DBLP:journals/toct/Ailon16}) if we replace each polynomial by a row-vector of its coefficients. 
Therefore, using the potential change bound  theorem of \cite{DBLP:journals/toct/Ailon16}, we  have the
statement of the theorem.  For self containment, we also provide a full proof in  Section~\ref{sec:proof:lem:entropychange}
in the appendix.
\end{proof}

\section{The preconditioners}\label{pppppp}

We now define our two polynomial matrix preconditioners, $A[z]$ and $B[z]$, that will be used to define
a potential for proving the main theorem.
First some extra notation: $\rho := \frac{\log \kappa}{2\log \Delta^{-1}}$, or $\Delta^{-\rho} = \sqrt \kappa$.  
Using the notation, we get that for all $t=0,\dots, m$, $ \deg(M^{(t)}_\Delta) - \val(M^{(t)}_\Delta) \leq 2\rho\ .$
This implies that either
\begin{enumerate}
\item[(i)]
$ \val(M^{(t)}_\Delta) \geq -\rho $, or
\item[(ii)]
$\deg(M^{(t)}_\Delta) \leq \rho$.
\end{enumerate}

We will assume the case (i) (a symmetric construction can be done for the case (ii), and hence there is
no loss of generality.)
We will also make the implicit assumption that
$\kappa=n^{o(1)}$, because otherwise the statement of the theorem is obvious from the trivial observation
that computation of $F$ requires linearly many steps.
For convenience, we will define $\mu=1-\Delta$, and 
let $\ell$ be the smallest integer such that $\Delta^\ell \leq 1/2$.  By our assumptions on
$\Delta$, this also implies $\Delta^\ell \geq 1/4$, hence $\Delta^\ell = \Theta(1)$.
Also note that
\begin{equation}\label{eq:ellapprox}
\ell = \Theta(1/\mu)\ .
\end{equation}

Let $d$ be $\deg(M^{(m)}_\Delta)$.  The preconditioner $A$ is defined as
$$A[z] = \Id(1+\Delta z^{-1} + \Delta^{2}z^{-2} +  \Delta^{3}z^{-3}+\cdots + \Delta^{\rho + d+1+\ell}z^{-\rho-d-1-\ell})\ .$$
Consider now the polynomial matrix $P[z] = M_\Delta^{(m)}[z]\cdot A[z]$.  
\begin{claim}\label{popopopopo}
 The (matrix) coefficient
of $P[z]$ corresponding to $z^{-\rho-i}$ is exactly $F\Delta^{i+\rho}$ for all $i$ in the range $R := \{0,1,\dots, \ell\}$.
\end{claim}
\begin{proof}
By polynomial multiplication definition, $\coeff{P}{-\rho-i}$ is
\begin{eqnarray}\sum_{k=-\rho-d-1-\ell}^0 \coeff{M_\Delta^{(m)}}{-\rho-i-k}\coeff{A}{k}
&=& \sum_{k=-\rho-d-1-\ell}^0 \Delta^{-k} \coeff{M_\Delta^{(m)}}{-\rho-i-k} \nonumber \\
= \Delta^{i+\rho}\sum_{k=-\rho-d-1-\ell}^0 \Delta^{-\rho-i-k} \coeff{M_\Delta^{(m)}}{-\rho-i-k}
&=&  \Delta^{i+\rho} M^{(m)}_\Delta[\Delta] =  \Delta^{i+\rho}  F\ ,
\end{eqnarray}
as required.
\end{proof}

Remember that for our construction, for all $i$ in the range $[0,\ell]$, $\Delta^{i+\rho} =  \Theta(\sqrt \kappa)$.  
We now define the other preconditioner, $B[z]$ as
$ B[z] := (M_{\Delta}^{(m)})^* \sum_{i=-\rho-\ell}^{-\rho} z^{i}F\ .$

We make a short note to demystify and provide intuition for our choice of  the two preconditioners.
 By the definition of $B$,
$M_\Delta^{(m)} B = \sum_{i=-\rho-\ell}^{-\rho} z^{i}F$. 
This implies that for all $i=-\rho-\ell,\dots, -\rho$, the coefficient  corresponding to $z^i$ of  both $M_\Delta^{(m)} A$ and of $(M_\Delta^{(m)})^* B$ are proportional to $F$. Thus we have aligned the two operands of the product inside $h(\cdot)$ in (\ref{yokyok}) so that the potential is high.
 At the same time, in the beginning of the computation, the potential is low because $A$ has only very few 
nonzeroes (namely, on the diagonal).   Creating this big potential gap, which we quantify in Lemma~\ref{lem:entropybounds} below, is  one of the requirements for our proof to work.  The other is upper bounding the  potential change at each step, as we now do.

\begin{lemma}\label{clm:boundnormAB}
For any paraunitary $M$, $ \|MA\|_\mynorm \leq \Theta(\sqrt \ell)$ and $ \|MB\|_\mynorm \leq \Theta(\sqrt \ell)\ .$
\end{lemma}
\begin{proof}
The proof relies on the fact that both matrices $MA$ and $MB$ belong to the family $\X$ defined as:
$$  \{U[z]p[z]X:\ U\mbox{ is paraunitary}, p\mbox{ is a polynomial in }\C[z], X\in \C^{n\times n}\mbox{ is unitary}\}\ .$$
The main technical claim, 
is to show that for  $Y = U[z]p[z]X\in \X$, any row of $Y$ has norm $\|p\|$, as can be seen by the following chain
for fixed $i\in [n]$:


\begin{eqnarray}
\|Y_{i,:}\|^2 &=& \sum_{j=1}^n\| Y_{i,j}\|^2 \nonumber 
=  \sum_{j=1}^n \int_{0}^1 |Y_{i,j}[e^{2\pi\iota t}]|^2dt \nonumber
=  \sum_{j=1}^n \int_{0}^1 Y_{i,j}[e^{2\pi\iota t}] \overline{ Y_{i,j}[e^{2\pi\iota t}] }dt \nonumber\\
&=&  \int_{0}^1   \sum_{j=1}^nY_{i,j}[e^{2\pi\iota t}] \overline{ Y_{i,j}[e^{2\pi\iota t}] }dt \nonumber
=  \int_{0}^1   \sum_{j=1}^n(U[z]p[z]X)_{i,j}[e^{2\pi\iota t}] \overline{ (U[z]p[z]X)_{i,j}[e^{2\pi\iota t}] }dt \nonumber\\
&=&  \int_{0}^1   \sum_{j=1}^n(U[z]p[z]X)_{i,j}[e^{2\pi\iota t}] \overline{ (U[z]p[z]X)_{i,j}}[e^{2\pi\iota t}] dt \nonumber \\
&=&  \int_{0}^1 \left (  \sum_{j=1}^n(U[z]p[z]X)_{i,j}\overline{ (U[z]p[z]X)_{i,j}}\right)[e^{2\pi\iota t}] dt \nonumber \\
&=&  \int_{0}^1 \left (  U[z]p[z]X)(U[z]p[z]X)^*\right)_{i,i}[e^{2\pi\iota t}] dt \nonumber
=  \int_{0}^1 \left (  U[z]p[z]X X^* p^*[z] U^*[z]\right)_{i,i}[e^{2\pi\iota t}] dt \nonumber\\
&=&  \int_{0}^1 \left (  U[z]p[z] p^*[z] U^*[z]\right)_{i,i}[e^{2\pi\iota t}] dt \nonumber
=  \int_{0}^1 \left ( p[z]p^*[z]\right)\left( U[z] U^*[z]\right)_{i,i}[e^{2\pi\iota t}] dt \nonumber\\
&=&  \int_{0}^1 \left ( p[z]p^*[z]\right)[e^{2\pi\iota t}] dt\nonumber 
=  \int_{0}^1 \left | p[e^{2\pi\iota t}]\right|^2 dt\label{conclusion} 
= \|p\|^2\  .
\end{eqnarray}
The first equality was by definition, the second by (\ref{parseval}), the fourth by changing order
of integration and summation, the fifth by the definition of $Y$, the sixth by our definition of polynomial
conjugation, the seventh by the fact that polynomial evaluation commutes with polynomial multiplication,
the eighth by definition of matrix multiplication, the ninth by properties of conjugation of products, the tenth by unitarity of $X$, the eleventh by commutativity of $p$ (a polynomial over scalars) and $P$ (a polynomial over matrices), the twelfth by paraunitarity of $U$, and the fourteenth again
by (\ref{parseval}).

Now, by our construction, $MA$ equals 
$U[z]p[z]X\in \X$ for $U=M$, $X=\Id$, and $$p[z] = 1+\Delta z^{-1} + \Delta^2 z^{-2} + \cdots \Delta^{\rho+d+1+\ell} z^{-\rho-d-1-\ell}\ .$$
   Using  the technical claim, 
\begin{eqnarray}
\|A\|_\mynorm = \|p \|&=& \sqrt{\sum_{k=0}^{\rho+d+1\ell} \Delta^{2k}} 
\leq  \sqrt{\frac{1}{1-\Delta^2}} =  \sqrt{\frac{1}{1-(1-\mu)^2}} = \Theta\left(\sqrt{1/\mu}\right) = \Theta(\sqrt{\ell})\ .\nonumber
\end{eqnarray}
As for $MB$, notice that  it equals  $U[z]p[z]X\in \X$ for $U=  M\cdot (M_\Delta^{(m)})^*$, $p= \sum_{i=\rho}^{\rho+\ell} z^{-i}$, $X=F$.   Therefore again we use the technical claim to get
$\|B\|_\mynorm = \|p \|= \sqrt{\sum_{k=\rho}^{\rho+\ell}  1}  = \Theta(\sqrt \ell)$.
\end{proof}
\noindent
Let us now compute the preconditioned quasi-entropy of $M_\Delta^{(0)} = \Id$ and of  $M_\Delta^{(m)}$.  \begin{lemma}\label{lem:entropybounds}
\begin{equation}
\Phi_{A,B}(M^{(m)}_\Delta) = \Omega\left (\frac{ \ell  n\log n}{\sqrt \kappa}\right) 
\ \ \ \ \ \ \ \ 
\Phi_{A,B}(\Id)                         = o\left(\frac{\ell n\log n}{\sqrt{\kappa}}\right )\ . \nonumber
\end{equation}
\end{lemma}
\begin{proof}
For the left bound in the lemma statement: 
\begin{eqnarray}
\Phi_{A,B}(M^{(m)}_\Delta) &=& 
\sum_{k=-\rho-\ell}^{-\rho}\sum_{i,j} h\left(\coeff{(M^{(m)}_\Delta A)_{i,j}}{k} \cdot \overline{\coeff{(M^{(m)}_\Delta B)_{i,j}}{k}}\right) \nonumber \\
 &=&
 \sum_{k=-\rho-\ell}^{-\rho}\sum_{i,j} h\left (( \Delta^{-k}F_{i,j})\cdot( F_{i,j})\right) \label{hphphphp}\\
 &=& \sum_{k=-\rho-\ell}^{-\rho}\sum_{i,j} h(F_{i,j}^2\Delta^{-k}) 
= -\sum_{k=-\rho-\ell}^{-\rho}\sum_{i,j} \Delta^{-k} F_{i,j}^2\log (\Delta^{-k} F_{i,j}^2)\ . \nonumber 
  \end{eqnarray}
  Recall that 
 $ \Delta^{-\rho} = \kappa^{1/2} = n^{o(1)}$.  
Also recall that  by definition of $\ell$, $\Delta^\ell = \Theta(1)$.  Therefore $\log \Delta^{-k}$ is $o(\log n)$ for all $k\in R$, and  
 \begin{eqnarray}
\Phi_{A,B}(M_\Delta^{(m)}) &=& -\sum_{k=-\rho-\ell}^{-\rho}\sum_{i,j} \left (\Delta^{-k} F_{i,j}^2\log F_{i,j}^2 
        +o(\log n) \Delta^{-k} F_{i,j}^2  \right) \nonumber \\
    &=&\sum_{k=-\rho-\ell}^{-\rho} \Delta^{-k} \Phi(F)  - o(\ell\Delta^{\rho}n\log n)
   = \Theta(\ell\Delta^\rho) \Phi(F) - o(\ell\Delta^{\rho}n\log n)  \nonumber \\
          &=& \Theta(\ell\Delta^\rho) n\log n  - o(\ell\Delta^{\rho}n\log n)  = \Omega(\ell\kappa^{-1/2} n \log n) \ , \nonumber
\end{eqnarray}
where the second equality  is from $\sum F_{i,j}^2 = n$ (unitarity of $F$) and definition of $\Phi$, the third is
from the fact that $\Delta^i = \Theta(1)$ for all $i=0\dots\ell$, and the fourth from (\ref{fourierentropy}).
For the right bound in the lemma statement: 
\begin{eqnarray}
\Phi_{A,B}(\Id) &=& 
\sum_{k}\sum_{i,j=1}^n h\left(\coeff{( A)_{i,j}}{k} \cdot\overline{\coeff{(B)_{i,j}}{k}}\right) \nonumber \\
 &=&
 \sum_{k=-\rho-\ell}^{-\rho}\sum_{i=1}^n h\left (( \Delta^{-k}1)\cdot\overline{ \coeff{B_{i,i}}{k}}\right) 
=
 \sum_{i=1}^n  \sum_{k=-\rho-\ell}^{-\rho}h\left (( \Delta^{-k}1)\cdot\overline{ \coeff{B_{i,i}}{k}}\right)\ . \nonumber \label{eq:PhiABId}
\end{eqnarray}
Fixing $i\in [n]$, we upper bound the absolute value of the last inner sum, which equals:
$$\underbrace{ \sum_{k=-\rho-\ell}^{-\rho}\Delta^{-k} h\left(\overline{ \coeff{B_{i,i}}{k}}\right)}_{E_1} - \underbrace{ \sum_{k=-\rho-\ell}^{-\rho}\Delta^{-k} \overline{ \coeff{B_{i,i}}{k}} \log \Delta^{-k}}_{E_2}\ .$$
To bound $E_1$, we first use  Lemma~\ref{clm:boundnormAB}, to argue that for all $i\in [n]$, $\|B_{i,:}\|^2 = \ell+1$.  (To be precise, it was  only shown that $\|B_{i,:}\|^2= \Theta(\ell)$ but it is obvious from the proof of Lemma~\ref{clm:boundnormAB} that the exact value is $\ell+1$.)  This in particular implies that
$\sum_{k=-\rho-\ell}^{-\rho}  |\coeff{B_{i,i}}{k}|^2 \leq \ell+1$.
 We therefore bound $E_1$ by:
\begin{equation}\label{qewrvgvg}
\sup_{\sum \tau_k^2  \leq \ell+1} \left | \sum_{k=-\rho-\ell}^{-\rho} \Delta^{-k} \tau_k \log |\tau_k| \right |\ .
\end{equation}
This is shown to be  $\Theta(\kappa^{-1/2}\ell)$, using standard calculus, as follows.  By the triangle inequality, (\ref{qewrvgvg}) is at most:
\begin{equation}\label{zxcbsreh}\sup_{\sum \tau_k^2  \leq \ell+1}  \sum_{k=-\rho-\ell}^{-\rho} \Delta^{-k} \left |\tau_k \log |\tau_k| \right |\ . \end{equation}
We now use the fact that $\Delta^{k} = \Theta(\Delta^\rho) = \Theta(\kappa^{-1/2})$ for all $k\in R=\{\rho,\dots, \rho+\ell\}$.  Therefore (\ref{zxcbsreh}) is at most a constant times
\begin{equation}\label{adsf9jf00000dsl}\kappa^{-1/2}\sup_{\sum \tau_k^2  \leq \ell+1}  \sum_{k=1}^{\ell+1}  \left |\tau_k \log |\tau_k| \right |\ . \end{equation}

The contribution coming from $k$'s such that $|\tau_k|\leq 1$ to  (\ref{adsf9jf00000dsl}) is at most
 $\Theta(\kappa^{-1/2}\ell)$, because $\left |\tau\log|\tau|\right |$ is a continuous function and hence bounded in the range $[-1,1]$.  
We hence bound
\begin{equation}\label{zxcbs99reh}
\sup_{K\leq \ell+1} \sup _{\left [ \begin{matrix} {\forall k\leq  K: \tau_k >1} \\ {\sum \tau_k^2\leq \ell+1}\end{matrix}\right ]} \kappa^{-1/2}\sum_{k=1}^K  \tau_k \log \tau_k\ . \end{equation}
Using the method of lagrange multipliers, the inner supremum  can only be obtained at one of the two points: (1) At the extreme candidate point $(\tau_1=\sqrt{\ell-K+2}, \tau_2=\dots=\tau_K=1)$, for which the function value is $O(\kappa^{-1/2} \sqrt\ell\log\ell)$, or  (2) when $(\tau_1=\dots=\tau_K=\sqrt{\frac{\ell+1}{K}})$, for which case the
value is $\kappa^{-1/2} \sqrt {K(\ell+1)}\log \sqrt{\frac{\ell+1}{K}}$.  But the last expression, maximized
over $K\in [1,\ell+1]$, is at most $\Theta(\ell)$, obtained at $K=\Theta(\ell)$.
Combining these arguments, the value of  (\ref{zxcbsreh}) is $\Theta(\kappa^{-1/2}\ell)$, as required.

To bound $E_2$, we use the fact that for all $k\in R$, $\Delta^{-k}=\Theta(\Delta^\rho)= \Theta(\kappa^{-1/2})$ and $\log \kappa=o(\log n)$.  Therefore $|E_2|$ is $o\left (\kappa^{-1/2}(\log n) \sum_{k=-\rho-\ell}^{-\rho} | \coeff{B_{i,i}}{k}|\right )$.  But by the well known $\ell_1/\ell_2$ bound,
we have that $\sum_{k=-\rho-\ell}^{-\rho} | \coeff{B_{i,i}}{k}| \leq \ell+1$.  Therefore, $|E_2|$ is $o(\kappa^{-1/2}\ell \log n)$. 
Combining, we have that $|\Phi_{A,B}(\Id)| = o(\ell \kappa^{-1/2}n \log n)$.
\end{proof}

Combining Lemma~\ref{lem:entropybounds} and~\ref{lem:entropychange} with Claim~\ref{clm:boundnormAB}, we get that the algorithm $\A$
increases the potential defined by the preconditioned quasi-entropy $\Phi_{A,B}(\cdot)$ from $o\left (\frac {\ell n \log n}{\sqrt \kappa}\right)$ to 
$ \Omega\left (\frac {\ell n \log n}{\sqrt \kappa}\right)$,
and at each step it can change the potential by no more than $\Theta(\sqrt \ell)\Theta(\sqrt \ell) = \Theta(\ell)$.  Hence, the number of steps is as stated in Theorem~\ref{thm:main}, concluding its proof.
(Reminder: See Section~\ref{sec:def:algcond}
for generalization to the case $\A$ is integral with respect to no $\Delta$.)

\section{Main Conjecture with Some Evidence, and Future Work}\label{evidence}
Our main conjecture, which is an algebraic version of Conjecture~\ref{conj1}, is as follows:
\begin{conjecture}\label{conj2}[Algebraic version of Conjecture~\ref{conj1}]
$\kappa$-Algebraic well conditioned computation allows only $O(\log(\kappa))$ speedup of FFT.
\end{conjecture}

There is interesting evidence for this conjecture, coming form recent work by Ailon and Yehuda \cite{DBLP:journals/corr/AilonY16}.
Recall that in our construction,
we started with an algorithm  $$\A=(\Id=M^{(0)},M^{(1)}, \dots, M^{(m)}=F)$$ working in $\R^n$, and
mapped  it to a pseudo-algorithm  $\A_\Delta$ working in $n$ dimensions over the ring of polynomial matrices.  The pseudo-algorithm
was used in the analysis only, and it was not assumed to be emulated.  However, such an emulation could be done.   We now explain
how to emulate $\A_\Delta$ using reals (and not polynomials), and then find surprising connections to other work.

 If all  polynomial
matrices in $\A_\Delta$ have monomials of degrees in the range $[-L, L]$, then they can be naturally embedded as real matrices of 
shape $n(2L+1)\times n(2L+1)$.  To explain the embedding, fix a matrix $M[z]$ in the ring of polynomials over $n\times n$ real matrices with monomials in the range $[-L,L]$ only.   The embedding will be denoted $\Psi(M) \in \R^{n(2L+1)\times n(2L+1)}$.
Divide the index set of size $n(2L+1)$  into $2L+1$ block-indices of size $n$ each.    The value of the block-entry $(i,j)$ of $\Psi(M)$ is $\coeff{X}{(j-i)}$, where the integer $j-i$ is understood to be modulo $2L+1$ in the branch $[-L,L]$.    It is not hard to see that this embedding is a homomorphism of the ring of polynomials modulo the identity $(z^{L+1}\equiv z^{-L})$,  over $n\times n$ matrices.

The algorithm $\Psi(\A) := \{ \Psi(M^{(0)}_\Delta), \dots, \Psi(M^{(m)}_\Delta)\}$ in $\R^{n(L+1)\times n(2L+1)}$ can be emulated
as follows.  A multiplication of  row $i$ of $M^{(t)}_\Delta$  by the monomial $z^{a}$ for an integer $a$ (giving $M^{(t+1)}_\Delta$) is emulated by performing a cyclic shift
of the set of $(2L+1)$ rows in $\Psi(M^{(t)}_\Delta)$.  This set of rows corresponds exactly to the  $i$'th index of all $(2L+1)$ index-blocks.  The cyclic shift can be done by a sequence $2L$   swaps, each  swap a rotation (a reflection, in fact).
A rotation affecting rows $i_1,i_2$ of  $M^{(t)}_\Delta$ is emulated by performing $2L +1$ rotations, affecting all pairs of rows corresponding to the pair $(i_1,i_2)$ within each of the $2L+1$ index-blocks.

It is easy to be convinced of the emulation correctness.  
 It is also easy to see that all matrices $\Psi(M^{(t)}_\Delta)$ are unitary (this property could have been proven directly from the definition of $\Psi$ and the para-unitarity of the $M^{(t)}_\Delta$'s).  To summarize, we started with an algorithm over general  $n\times n$ matrices making $m$ steps, passed through an algorithm over paraunitary $n\times n$ polynomial matrices making $m$ steps, and arrived
at an algorithm over unitary $n(2L+1)\times n(2L+1)$ matrices making $\Theta(L m)$ steps.

\begin{eqnarray*}
M^{(t)} \ \ \ \ \ \ \ \ \ \mapsto & M^{(t)}_\Delta[z] &\ \ \ \ \ \  \hookrightarrow \ \ \ \ \ \   \ \ \Psi(M^{(t)}_\Delta[z]) \\
\mbox{\parbox{2.2cm}{\center{Original \\ algorithm \\ over reals \\ $n\times n$ \\ non-unitary}}}\ \ \ \ \  \ \ \ \ \ 
&
\mbox{\parbox{4cm}{\center{Abstract \\ algorithm \\ over  polynomials \\ $n\times n$  \\ paraunitary}}}
&
\mbox{\ \ \ \ \ \ \parbox{5cm}{\center{Emulation of \\ abstract  algorithm \\ over reals \\ $n(2L+1)\times n(2L+1)$ \\ unitary}}}
\end{eqnarray*}

Now, what can be said about the final matrix in the emulation, $\Psi(M^{(m)}_\Delta)$?   We only care about the case $m=o(n\log n)$ of course, from which
we have the entropy bound  $$\Phi(\Psi(M^{(m)}_\Delta)) = o(Ln\log(L n))\ .$$ (This can be derived from the definition (\ref{def:entropy}) together with  Lemma~\ref{lem:entropychange}, but it is recommended to use \cite{Ailon13} for a simpler treatment of unitary computation).

Take $L$ to be $\rho+\ell$, as defined in Section~\ref{pppppp}, and assume that for all $t=0..m$, $-\rho \leq \val(M^{(t)}_\Delta) \leq \deg(M^{(t)}_\Delta) \leq \rho$.  (This is a simplifying assumption that can be removed using a simple trick that we omit.)
From Claim~\ref{popopopopo} we have that the matrix  $\Psi(M^{(m)}_\Delta)\Psi(A)$ (where $A$ is the preconditioner from
Section~\ref{pppppp}) contains $\Theta(\rho)$ copies of the matrix $F\Theta(1/\sqrt \kappa)$ as blocks.\footnote{Note that $\Psi(M^{(m)}_\Delta)\Psi(A)$ is not the same as $\Psi(M^{(m)}_\Delta A)$, due to modular wrap-around of monomials.  But the argument is still correct.}  
 This is an interesting outcome:  We have a matrix $\Psi(M^{(m)}_\Delta)\Psi(A)$ of total low entropy (because $m$ is small) which contains $\Theta(\rho)$ copies of  $F$, each scaled down by $\Theta(\sqrt \kappa)$.\footnote{We also need to show that the preconditioner $\Psi(A)$ can affect the
entropy of $\Psi(M^{(m)}_\Delta)$ by at most $o(LN\log (LN))$, which we omit here for the sake of simplicity, and since this section is used for explaining intuitive evidence.}  This is an odd situation.  We have computed, in our emulation, and using an algorithm of perfect condition number $1$,  a matrix of low entropy  that contains (under a low-entropy preconditioning)  a scaled-down high-entropy matrix as sub-matrices.  What is the computational complexity of such a matrix?

 Ailon and Gal tackled this question  in \cite{DBLP:journals/corr/AilonY16},  
by studying a toy example:   Is it possible to compute  a ``Fourier Perturbation": $$\Id+\eps F\ ,$$ for some small $\eps$ using an almost unitary algorithm? (Note that
the matrix $\Id+\eps F$ is not quite unitary, but very close.) The answer is that at least $\Omega((n\log n)/\log(1/\eps))$ steps
are needed in a computational model very close to unitary.     The parameter $\eps$ corresponds to our $1/\sqrt{\kappa}$.  If we could somehow generalize that result here, we  would get a denominator of $\log \kappa$ in (\ref{gugugu}).    Doing that seems to require exploring very interesting derivations of entropy
functions, as Ailon and Yehuda did in  \cite{DBLP:journals/corr/AilonY16}  for the toy problem.  We were unable as of yet to obtain this generalization.

By Lemma~\ref{lem:condition} the following conjecture is stronger than Conjecture~\ref{conj2}, and therefore harder to prove:
\begin{conjecture}\label{conj3}[Geometric version of Conjecture~\ref{conj1}]
$\kappa$-(geometric) well conditioned computation allows only $O(\log(\kappa))$ speedup of FFT.
\end{conjecture}
It is interesting to study algebraic condition number in its own right as a measure of  numerical stability.  Consider the following  related question:  Lemma~\ref{lem:condition} tells us that if an algorithm is algebraically $\kappa$ well conditioned then it is also  (geometrically) $\kappa$-well conditioned.  The converse is clearly not true, but perhaps some weak form of the converse should hold.  In the extreme
case, for example, note that an algorithm that is (geometrically) $1$-well conditioned \emph{must also be} algebraically $1$-conditioned, because such an algorithm can only use constants $1$ and $-1$.
Could  we generalize this observation to obtain some kind of converse for Lemma~\ref{lem:condition}?

\noindent
Some more noteworthy open problems: 
\begin{itemize}
\item
 Extending the results to the extra-memory (non in-place) model. 
\item  Obtaining computational lower bounds for computing a Johnson-Lindenstrauss transform.  There has been work showing that a JL transform can be sped up using the Fourier transform \cite{DBLP:journals/siamcomp/AilonC09}, but computational lower bounds
for Fourier do not imply computational lower bounds for JL.  Note that there are tight lower bounds
for dimensionality parameters of JL \cite{DBLP:conf/icalp/LarsenN16}, but they are not computational. 
We also note the question of whether this line of work can be used to show lower bounds for the \emph{sparse} Fourier transform \cite{DBLP:conf/focs/IndykK14, DBLP:conf/soda/IndykKP14, 6879613}.
\end{itemize}
\bibliography{low_bound_fft} 

\appendix

\section{The General Case: Algorithms That are Not Integral}\label{sec:def:algcond}

In this section we show how to extend the main result in this work to the case in which the
algorithm $\A$ is integral with respect to no $\Delta$.

\subsection{Extension of `Well Conditioned' Definition to the General Case}
Fix an algorithm $\A = (\Id = M^{(0},\dots, M^{(m)})$ in $\R^n$.  Let $$\Delta_1 < \Delta_2 < \cdots < \Delta_q <\dots$$ be  an infinite,  increasing sequence of numbers in the open interval $\left (\frac 2 3, 1\right )$, tending to $1$ in the limit.  For each index $q$ we define another algorithm $\A_q = (M^{(0}_q,\dots, M^{(m)}_q)$ inductively, as follows.  First, $M^{(0)}_q = \Id$.  For $t\geq 1$, $M^{(t)}_q$ is obtained from $M^{(t-1)}_q$ depending on whether $M^{(t)}$ is obtained from $M^{(t-1)}$ by a rotation or a constant
gate.  In case of rotation, $M^{(t)}_q$ is obtained by applying the same rotation to $M^{(t-1)}_q$.
In case of a constant gate $c$ acting on row $i$, The matrix $M^{(t)}_q$ is obtained from $M^{(t-1)}_q$ 
by multiplying the $i$'th row by
$$ c_q := \Delta_q^ {\lceil \log_{\Delta_q} c\rceil}\ .$$

The resulting algorithm $\A_q$ is clearly integral with respect to $\Delta_q$.
Additionally, it is not hard to see by standard limit calculus and by induction, that
tor each $t=0\dots m$, $\lim_{q\rightarrow \infty} M^{(t)}_q = M^{(t)}$.
We are now ready to extend Definition~\ref{defn:algwellcond} to the non integral case.
\begin{definition}\label{tototototo}
$\A$ is algebraically $\kappa$-well conditioned if for some sequence $\Delta_1<\Delta_2<\cdots$ as above,
there exists a sequence $\kappa_1, \kappa_2,\dots$ tending to $\kappa$,  such that  $\A_q$ is
$\kappa_q$-well conditioned for all  $q$ (per Definition~\ref{defn:algwellcond}).  
\end{definition}

It should be noted that if $\A$ is algebraically $\kappa$-well conditioned, then any sequence $\Delta_1<\Delta_2<\cdots$ (and not just one) will serve as a witness for the definition.  It should also be noted that  Definition~\ref{tototototo} coincides with Definition~\ref{defn:algwellcond}  for algorithms $\A$
that are $\Delta$-integral for some $\Delta$.  This can be seen, for example, by choosing $\Delta_q = \Delta^{1/q}$
and $\kappa_q=\kappa$ for all $q$.

\subsection{Proof of Lemma~\ref{lem:condition} for The General Case}\label{sec:lemgen}

Assume $\A=(M^{(0)},\dots, M^{(m)})$ is algebraically $\kappa$-well conditioned (per the general Definition~\ref{tototototo}).  Let $\Delta_1<\Delta_2<\cdots$ and $\kappa_1,\kappa_2,\dots$ be as guaranteed to exist by the definition.  Then by Lemma~\ref{lem:condition} (established for the integral case), the algorithm $\A_q$ is geometrically
$\kappa_q$-well conditioned for all $q$.  But the  geometric condition number of a matrix is  a continuous invariant in the domain of nonsingular matrices, because both spectral norm and matrix inverse are continuous functions over matrices.  Therefore, for all $t=0\dots m$, $M^{(t)}$ (which is also the limit of $M^{(t)}_q$) is geometrically $\kappa$-well conditioned.  

\subsection{Proof of Theorem~\ref{thm:main} for the General Case}\label{sec:thmgen}
Assume $\A=(M^{(0)},\dots, M^{(m)})$ is algebraically $\kappa$-well conditioned (per the general Definition~\ref{tototototo}).  
 Let $\Delta_1<\Delta_2<\cdots$ and $\kappa_1,\kappa_2,\dots$ be as guaranteed to exist by the definition.
 Let $F_q = M^{(m)}_q$, namely, the resulting matrix computed by $\A_q$.   First we note that  $\lim F_q = F$.
For each $q$,  we now want to use the statement of the theorem in the integral case, which we have proved.  The problem is that
the resulting matrix $F_q$ does not necessarily have the two key  properties of Fourier transform used in the theorem:
(1) unitarity and (2) high potential $\Phi(F_q)$. 

In order to adjust the proof for  the case of $\Delta_q$-integral algorithm $\A$ computing $F_q$ (instead of $F$), we will first  assume that $q$ is large enough so that
\begin{equation}\label{ooooo} \forall i,j\in[n]\ \ \ |(F_q)_{i,j}-F_{i,j}| \leq |F_{i,j}|/2\ .\end{equation}
This is achievable, because $F_{i,j}\neq 0$ for all $i,j$, for both types of Fourier transform.\footnote{We could also make due with zero valued matrix elements, but for simplicity we avoid this technicality here.}
In particular, (\ref{ooooo}) implies:
\begin{equation}\label{iiiiiii} \sum_{i,j} F_{i,j} (F_q)_{i,j} = \Theta(n)\ \ \ \ \ \Phi_F(F_q) := -\sum_{i,j} F_{i,j}(F_q)_{i,j} \log  (F_{i,j}(F_q)_{i,j}) \geq \Omega(n \log n)\ .\end{equation}
The trick is now to use preconditioners $A_q, B_q$ as defined in Section~\ref{pppppp}.
$$ A_q = \Id(1+\Delta_q z^{-1} + \Delta_q^{2}z^{-2} +  \Delta_q^{3}z^{-3}+\cdots + \Delta_q^{\rho_q+\ell_q}z^{-\rho_q-\ell_q})\ \  B_q = ((M_q)_{\Delta_q}^{(m)})^* \sum_{i=\rho_q}^{\rho_q+\ell_q} z^{-i}F\ ,$$
where $\rho_q$ is $\frac{\log \kappa_q}{2\log \Delta_q^{-1}}$ and $\ell_q$ is  the smallest integer such that
$\Delta_k^\ell \leq 1/2$.  Note that in the definition of $B_q$ we used $F$ and not $F_q$.  This is important, because we need Claim~\ref{clm:boundnormAB} and Lemma~\ref{lem:entropybounds} to work (they rely on $B$ being a product of a  unitary matrix, a paraunitary matrix and a polynomial over scalars).
Lemma~\ref{lem:entropybounds} now reads as
\begin{eqnarray}
\Phi_{A_q,B_q}((M_q)^{(m)}_\Delta) &=& \Omega\left (\frac{ \ell_q  n\log n}{\sqrt {\kappa_q}}\right) \label{potential bounds1} \\
\Phi_{A_q,B_q}(\Id)                         &=& o\left(\frac{\ell n\log n}{\sqrt{\kappa_q}}\right ) \label{potential bounds2}
\end{eqnarray}
For the proof of Lemma~\ref{lem:entropybounds}, the only thing worth noting is that the expression in line  (\ref{hphphphp}) becomes 
$$ \sum_{k=-\rho_q-\ell_q}^{-\rho_q}\sum_{i,j} h\left (( \Delta_q^{-k}(F_q)_{i,j})\cdot( F_{i,j})\right)\  , $$
that is, one copy of $F_{i,j}$ (coming from the preconditioner $A_q$) is replaced by $(F_q)_{i,j}$, and
the other remains untouched.  The remainder of the proof of the lemma proceeds in an obvious way, taking advantage of  (\ref{iiiiiii}).

The resulting lower bound for the number $m$ of steps that the algebraically $\kappa_q$-well conditioned algorithm $\A_q$ needs to compute $F_q$ is
$$C\frac{\Phi_F(F_q)}{\sqrt \kappa_q}\ ,$$
for sufficiently large $q$, for a global $C$.  Taking $q$ to infinity and recalling that $\kappa_q\rightarrow \kappa$ we conclude Theorem~\ref{thm:main}'s  proof for the general case.

\section{Proof of Lemma~\ref{lem:entropychange}}\label{sec:proof:lem:entropychange}
Without loss of generality, assume the rotation is applied to rows $1$ and $2$ of $M$, so
\begin{equation}\label{gtrgtrgttgtgtgttgt} M' = \left (\begin{matrix}\cos \theta & \sin \theta & & & \\ -\sin \theta & \cos \theta & & & \\   &  & 1 & & \\ & & & \ddots & \\ & & & & 1 \end{matrix} \right ) M\  \end{equation}
for some $\theta \in [0,2\pi)$.\footnote{To be technically precise, we also need to take into consideration reflections.  But these can be composed as a product of a matrix as in (\ref{gtrgtrgttgtgtgttgt}) and a matrix with $\pm 1$ on the diagonal.  It is trivial that such a matrix makes no difference to the potential.}
  It suffices to consider only the change in the contribution of rows $1$ and $2$ to the potential.   The contribution to the pre-rotation potential coming from the first two rows, which we denote $\Phi_{1,2}$, is 
$$\Phi_{1,2} = \sum_k  \sum_{j=1}^n \sum_{i=1}^2 h\left ( \coeff{(MA)_{i,j}}{k}\overline{\coeff{(MB)_{i,j}}{k}}\right )  .$$
Similarly, the contribution  to the post-rotation potential coming from the first two rows, denoted $\Phi'_{1,2}$, is
$$\Phi'_{1,2} = \sum_k  \sum_{j=1}^n \sum_{i=1}^2 h\left ( \coeff{(M'A)_{i,j}}{k}\overline{\coeff{(M'B)_{i,j}}{k}}\right )  .$$
By the triangle inequality:
$$ |\Phi_{1,2} - \Phi'_{1,2}| \leq \sum_k   \sum_{j=1}^n  \left | \underbrace{\sum_{i=1}^2h\left ( \coeff{(MA)_{i,j}}{k}\overline{\coeff{(MB)_{i,j}}{k}}\right ) -h\left ( \coeff{(M'A)_{i,j}}{k}\overline{\coeff{(M'B)_{i,j}}{k}}\right ) }_{\delta_{k,j}}\right |\ .$$
To avoid clutter, let  $$\alpha_{k,i,j} = \coeff{(MA)_{i,j}}{k}, \beta_{k,i,j}=\coeff{(MB)_{i,j}}{k}, \alpha'_{k,i,j} = \coeff{(M'A)_{i,j}}{k}, \beta'_{k,i,j}=\coeff{(M'B)_{i,j}}{k}\ .$$
For each $k,j$, let
\begin{eqnarray}
 \rho_{k,j} &:=& \sqrt{ |\alpha_{k,1,j}|^2 +  |\alpha_{k,2,j}|^2} =  \sqrt{ |\alpha'_{k,1,j}|^2 +  |\alpha'_{k,2,j}|^2} \label{yty} \\
  \sigma_{k,j} &:=& \sqrt{ |\beta_{k,1,j}|^2 +  |\beta_{k,2,j}|^2} =  \sqrt{ |\beta'_{k,1,j}|^2 +  |\beta'_{k,2,j}|^2} \label{tyt}\ , 
\end{eqnarray}
where the right hand equality in both (\ref{yty}) and (\ref{tyt}) is  from the fact that a rotation preserves inner products (and in particular norms) in $\C^2$.  By the definition of $h$, 
\begin{eqnarray}
h(\alpha_{k,1,j}\overline{\beta_{k,1,j}}) &=& -\alpha_{k,1,j}\overline{\beta_{k,1,j}}\log|\alpha_{k,1,j}\overline{\beta_{k,1,j}}|  \nonumber \\
&=&{\alpha_{k,1,j}\overline{\beta_{k,1,j}}}\log\left |  \rho_{k,j}\sigma_{k,j}\frac{\alpha_{k,1,j}\overline{\beta_{k,1,j}}}{\rho_{k,j}\sigma_{k,j}}  \right |  \nonumber \\
&=& \underbrace{-{\alpha_{k,1,j}\overline{\beta_{k,1,j}}} \log\left |  \rho_{k,j}\sigma_{k,j}\right | }_{(*)} \underbrace{-\rho_{k,j}\sigma_{k,j}\frac{\alpha_{k,1,j}\overline{\beta_{k,1,j}}} {\rho_{k,j}\sigma_{k,j}} \log\left | \frac{\alpha_{k,1,j}\overline{\beta_{k,1,j}}}{\rho_{k,j}\sigma_{k,j}}  \right |}_{(**)} \nonumber
\end{eqnarray}
where we simply multiplied and divided by $\rho_{k,j}\sigma_{k,j}$ inside the logarithm to get the second line, then used logarithm properties and  multiplied and divided by $\rho_{k,j}\sigma_{k,j}$ to get the third.  Similarly:
\begin{eqnarray}
h(\alpha_{k,2,j}\overline{\beta_{k,2,j}})
&=&\underbrace{- {\alpha_{k,2,j}\overline{\beta_{k,2,j}}} \log\left |  \rho_{k,j}\sigma_{k,j}\right | }_{(*)} \underbrace{-{\rho_{k,j}\sigma_{k,j}}\frac{\alpha_{k,2,j}\overline{\beta_{k,2,j}}}{\rho_{k,j}\sigma_{k,j}} \log\left | \frac{\alpha_{k,2,j}\overline{\beta_{k,2,j}}}{\rho_{k,j}\sigma_{k,j}}  \right | }_{(**)}\nonumber \\
h(\alpha'_{k,1,j}\overline{\beta'_{k,1,j}})
&=& \underbrace{- {\alpha'_{k,1,j}\overline{\beta'_{k,1,j}}} \log\left |  \rho_{k,j}\sigma_{k,j}\right | }_{(*)} \underbrace{-{\rho_{k,j}\sigma_{k,j}}\frac{\alpha'_{k,1,j}\overline{\beta'_{k,1,j}}} {\rho_{k,j}\sigma_{k,j}}\log\left | \frac{\alpha'_{k,1,j}\overline{\beta'_{k,1,j}}}{\rho_{k,j}\sigma_{k,j}}  \right |}_{(**)} \nonumber \\
h(\alpha'_{k,2,j}\overline{\beta'_{k,2,j}})
&=& \underbrace{-{\alpha'_{k,2,j}\overline{\beta'_{k,2,j}}} \log\left |  \rho_{k,j}\sigma_{k,j}\right | }_{(*)} \underbrace{-{\rho_{k,j}\sigma_{k,j}}\frac {\alpha'_{k,2,j}\overline{\beta'_{k,2,j}}}{\rho_{k,j}\sigma_{k,j}} \log\left | \frac{\alpha'_{k,2,j}\overline{\beta'_{k,2,j}}}{\rho_{k,j}\sigma_{k,j}}  \right | }_{(**)}\nonumber
\end{eqnarray}

Now to estimate $\delta_{k,j} = h(\alpha_{k,1,j}\overline{\beta_{k,1,j}})+h(\alpha_{k,2,j}\overline{\beta_{k,2,j}}) - h(\alpha'_{k,1,j}\overline{\beta'_{k,1,j}})-h(\alpha'_{k,2,j}\overline{\beta'_{k,2,j}})$, we first note that
the contribution coming from the (*) expressions vanishes.  Indeed, rotation preserves complex inner product, and hence $${\alpha_{k,1,j}\overline{\beta_{k,1,j}}} + {\alpha_{k,2,j}\overline{\beta_{k,2,j}}} = {\alpha'_{k,1,j}\overline{\beta'_{k,1,j}}} + {\alpha'_{k,2,j}\overline{\beta'_{k,2,j}}}\ .$$

To estimate the contribution coming from the (**) expressions, we simply note that the four fractions $\frac{\alpha_{k,1,j}\overline{\beta_{k,1,j}}}{\rho_{k,j}\sigma_{k,j}}, \frac{\alpha_{k,2,j}\overline{\beta_{k,2,j}}} {\rho_{k,j}\sigma_{k,j}} , \frac{\alpha'_{k,1,j}\overline{\beta'_{k,1,j}}}{\rho_{k,j}\sigma_{k,j}}, \frac{\alpha'_{k,2,j}\overline{\beta'_{k,2,j}}}{\rho_{k,j}\sigma_{k,j}}$ are at most $1$ in absolute value.  This implies that $|\delta_{k,j}| = O(\rho_{k,j}\sigma_{k,j})$.  Summing up over all $k$ and $j$, we get
\begin{eqnarray}
 \sum_{k,j} |\delta_{k,j}| &\leq& O\left ( \sum_{k,j} \rho_{k,j}\sigma_{k,j} \right)  
\leq O\left (\sqrt{\left( \sum_{k,j} \rho_{k,j}^2\right)\left(\sum_{k,j}\sigma_{k,j}^2 \right)}\right)  \nonumber \\
&=& O\left (\sqrt{\left(  \| (MA)_{1,:} \|^2+ \| (MA)_{2,:} \|^2  \right)\left(  \| (MB)_{1,:} \|^2+ \| (MB)_{2,:} \|^2  \right)}\right)  \nonumber \\
&=& O( \|MA\|_{2,\infty}\cdot\|MB\|_{2,\infty})\ . \nonumber
\end{eqnarray}
This concludes the proof.

\end{document}